\documentclass{article}
\usepackage{amsfonts,amsmath,amsthm,amssymb,authblk,latexsym,bm}
\usepackage{hyperref} \hypersetup{colorlinks=true} 
\usepackage{citeref}

\newtheorem{theorem}{Theorem}[section]
\newtheorem{corollary}[theorem]{Corollary}

\newtheorem{lemma}[theorem]{Lemma}

\newtheorem{example}[theorem]{Example}

\newtheorem{remark}[theorem]{Remark}

\numberwithin{equation}{section}

\def\vph{\varphi} 

\begin{document}
\title{Self-dual $2$-quasi-cyclic Codes and Dihedral Codes}
\author{
Yun Fan\par
{\small School of Mathematics and Statistics}\par\vskip-1mm
{\small Central China Normal University, Wuhan 430079, China}
\par\vskip3mm
Yuchang Zhang\par
{\small College of Mathematics and Statistics}\par\vskip-1mm
{\small Hubei Normal University, Huangshi 435002, China}
}
\date{}
\maketitle

\insert\footins{\footnotesize{\it Email address}:
yfan@mail.ccnu.edu.cn (Yun Fan);
}

\begin{abstract}
We characterize the structure of $2$-quasi-cyclic codes over a finite field~$F$ 
by the so-called Goursat Lemma.
With the characterization, we exhibit a necessary and sufficient condition for
a $2$-quasi-cyclic code being a dihedral code. And we obtain 
a necessary and sufficient condition for a self-dual $2$-quasi-cyclic code
being a dihedral code (if ${\rm char} F =2$), 
or a consta-dihedral code (if ${\rm char} F\ne 2$).
As a consequence, any self-dual $2$-quasi-cyclic code generated by one element
must be (consta-)dihedral. In particular, 
any self-dual double circulant code must be (consta-)dihedral.
Also, we show a necessary and sufficient condition that the three classes 
(the self-dual double circulant codes,  
 the self-dual $2$-quasi-cyclic codes,  
 and the self-dual (consta-)dihedral codes)   
are coincide each other.

\end{abstract}

\section{Introduction}

In this paper $F$ is always a finite field with cardinality $|F|=q$
which is a power of a prime. For any set $S$ we denote $|S|$ 
the cardinality of $S$. 
And $n>1$ is always an integer coprime to $q$. 

Any linear subspace $C$ of $F^n$, denoted $C\le F^n$, 
is called a linear code of length $n$ over $F$. 
in $F^n$ the inner product of $a=(a_0, a_1,\cdots,a_{n-1})$ and
$b=(b_0,b_1,\cdots,b_{n-1})$ is defined to be 
$\langle a,b\rangle=\sum_{i=0}^{n-1}a_ib_i$. 
Then the self-orthogonal codes, self-dual codes, LCD codes etc. 
are defined as usual, e.g., cf. \cite{HP}.

An $n\times n$ matrix over $F$ is said to be {\em circulant} 
if any row of the matrix is the cyclic shift of the previous row. 
An $n\times 2n$ matrix over $F$ is said to be {\em double circulant} 
if it is a concatenation side by side $(A~B)$ of two 
$n\times n$ circulant matrices $A$ and $B$. If $A$ is of full-rank, 
then $A^{-1}\!(A~B)=(I~A^{\!-1}\!B)$, where $I$ denotes the identity matrix 
and $A^{-1}\!B$ is still a circulant matrix.

A linear code $C\le F^{2n}$ is said to be 
{\em double circulant} if $C$ has a double circulant matrix $(I~A)$
as its generating matrix, e.g., see \cite{AOS}, \cite{TTHAA}. 

A linear code $C\le F^{n}$ is called a {\em cyclic code}
if it is invariant by the following cycle permutation on bits: 
$$
(a_0,a_1,\cdots,a_{n-1})
\mapsto (a_{n-1},a_0,\cdots,a_{n-2}).
$$  
A linear code $C\le F^{2n}$ is said to be 
{\em quasi-cyclic of index $2$}, {\em $2$-quasi cyclic} for sort, 
if~$C$ is invariant by the following two-cycle permutation on bits
$$
(a_0,a_1,\cdots,a_{n-1},\;b_0,b_1,\cdots,b_{n-1})
\mapsto (a_{n-1},a_0,\cdots,a_{n-2},\;b_{n-1},b_0,\cdots,b_{n-2}).
$$  
So double circulant codes are 2-quasi-cyclic codes; but 
the converse is not true.

In this paper we always assume that 
$H=\{1,x,\cdots,x^{n-1}\}$ with $x^n=1$ is a cyclic group of order $n$.
By $FH$ we denote the group algebra, i.e., $FH$ is an $F$-vector space
with basis $H$ and with multiplication induced by the group multiplication of $H$. 
So, $FH\cong F[X]/\langle X^n-1\rangle$. 
We identify any element $a=a(x)=a_0+a_1x+\cdots+a_{n-1}x^{n-1}\in FH$
 with the word $(a_0,a_1,\cdots,a_{n-1})\in F^n$.
Then the ideals of $FH$
(i.e., $FH$-submodules of the regular module $FH$) 
are just cyclic codes of length $n$ over~$F$.
Further, the outer direct sum $(FH)^2:=FH\times FH=\{(a,b)\,|\,a,b\in FH\}$ 
is an $FH$-module. Any $FH$-submodule of $(FH)^2$ is clearly 
a $2$-quasi cyclic code; and vice versa.
 
Let $T=\{1,y\}$ be a cyclic group of order $2$, 
and $G=H\rtimes T$ be the semidirect product with relation $yxy^{-1}=x^{-1}$;
i.e., $G$ is a dihedral group of order $2n$ with the cyclic subgroup $H$ of index $2$.
Then we have the group algebra $FG$ which is not commutative. 
Any left ideal of $FG$ ($FG$-submodule of the left regular module $FG$)  
is called a {\em dihedral code} of length $2n$ over $F$.

The dihedral group algebra $FG$ has $FH$ as a subalgebra. 
As $FH$-modules, $F\!G\!=\!F\!H\oplus F\!H y\cong F\!H\times F\!H$, 
where $F\!H\oplus F\!H y=\{a(x)+b(x)y\,|\,a(x),b(x)\in F\!H\}$ 
 denotes the inner direct sum. 
Therefore, any dihedral code is a 2-quasi cyclic code.  
The converse is clearly incorrect.

Recently, \cite{AOS} exhibits an interesting result: 
if the characteristic ${\rm char}F=2$, any self-dual double circulant code 
is a dihedral code. 
In fact, we'll show that it is still true for odd characteristic 
provided ``dihedral code'' is replaced  by
``consta-dihedral code'' (see Corollary \ref{principal self-dual odd} below). 

However, self-dual $2$-quasi-cyclic codes are not dihedral in general.
On the other hand, there are self-dual $2$-quasi-cyclic codes which are dihedral codes, 
but not double circulant codes.   

\begin{example}\label{counterexample}\rm 
Take $F=F_4=\{0,1,\omega,\omega^2\}$, where $1+\omega+\omega^2=0$. 
Take $n=3$,  $H=\{1,x,x^2\}$ with $x^3=1$, In $FH$ we take 
$$
e_0(x)=1+x+x^2, ~~~ e_1(x)=1+\omega x+\omega^2 x^2, ~~~ 
e_1'(x)=1+\omega^2 x+\omega x^2.
$$
In $FH\times FH$ we set
\begin{equation*}\begin{array}{l}
\widehat e_0=\big(e_0(x),\,e_0(x)\big)=(1,1,1,~1,1,1), \\
 g_1=\big(e_1(x),\,0\big)=(1,\omega,\omega^2,~0,0,0),  \\
 g_2=\big(0,\,e_1(x)\big)=(0,0,0,~1,\omega,\omega^2),\\
 g'_2=\big(0,\,e_1'(x)\big)=(0,0,0,~1,\omega^2,\omega).
\end{array}\end{equation*}

(1). 
Let $C=FH\widehat e_0+FHg_1+FHg_2$, which is an 
$F\!H$-submodule of $F\!H\!\times\! F\!H$.
It is easy to check that $\widehat e_0,g_1,g_2$ form an $F$-basis of $C$, hence
$$
 A=\begin{pmatrix}
  1&1&1&1&1&1\\
  1&\omega&\omega^2&0&0&0\\
  0&0&0&1&\omega&\omega^2
\end{pmatrix}
$$
is a generating matrix of $C$. Since 
$AA^T=0$ where $A^T$ denotes the transpose of $A$, $C$ is self-dual. 
Let $T=\{1,y\}$,  $G=H\rtimes T$ be the dihedral group of order~$6$ and
$FG=\{(a(x)+b(x)y\,|\,a(x),b(x)\in FH)\}$ as above. 
Thinking of that $\widehat e_0=e_0(x)+e_0(x)y$, 
$g_1=e_1(x)+0y$ and $g_2=0+e_1(x)y$, 
we can view $C$ as an $F$-subspace of $FG$. 
But $C$ is not an $FG$-submodule (left ideal) of $FG$,
because $yg_1=ye_1(x)=(1+\omega^2x+\omega x^2)y=e_1'(x)y\notin C$; 
i.e., $C$ is not dihedral. 

(2). 
 $C=FH\widehat e_0+FHg_1+FHg'_2$ is an $FH$-submodule of $FH\times FH$.
And $\widehat e_0,g_1,g'_2$ are an $F$-basis of $C$.
The following $A$ is a generating matrix of $C$: 
$$
 A=\begin{pmatrix}
  1&1&1&1&1&1\\
  1&\omega&\omega^2&0&0&0\\
  0&0&0&1&\omega^2&\omega
\end{pmatrix}.
$$
Since $AA^T=0$,  $C$ is a self-dual $2$-quasi-cyclic code.
This time, it is easy to check that $C$ is a dihedral code, 
but not a double circulant code 
(see Example~\ref{even examples}(2) for more explanations).
\end{example}

We are concerned with the exact relationship between the three classes of codes:
double circulant codes, $2$-quasi-cyclic codes and dihedral codes.

Question 1: what $2$-quasi-cyclic codes are double circulant codes? 

Question 2: what $2$-quasi cyclic codes are dihedral codes?

\medskip\noindent
And a special concern is about the self-dual ones
(e.g., \cite{AOS}, \cite{AGOSS}, \cite{MW},\cite{Musa}):

Question 3: what self-dual $2$-quasi-cyclic codes are (consta-)dihedral codes? 

Question 4: in what case (about $n$ and $q$) any self-dual $2$-quasi-cyclic code
is a (consta-)dihedral code? 

\medskip
The key idea of this paper is to characterize 
$2$-quasi-cyclic codes by the so-called Goursat Lemma; and apply
the characterization to solve these questions. 

In Section \ref{preliminaries}, we sketch preliminaries about $FH$. 
In Section~\ref{2-quasi-cyclic}, with Goursat Lemma we characterize 
the structure of $2$-quasi-cyclic codes 
(Theorem \ref{Goursat C} and its corollaries),
and answer Questions 1 and Question 2 (in Theorem~\ref{generator matrix}
and Theorem~\ref{2-quasi dihedral}).
Section~\ref{self-dual 2-quasi-cyclic} is devoted to 
a characterization of the self-duality of $2$-quasi-cyclic codes 
(Theorem~\ref{self-dual C}).
In Section~\ref{self-dual dihedral}, 
assuming that the characteristic of $F$ is even,
we answer Question~3 in Theorem~\ref{t char even}, 
and answer Question~4 by Theorem~\ref{c char even}, 
which (and Theorem~\ref{c char odd} below) are interestingly 
related to a classical question ``in what case any cyclic code is LCD?''
Finally, in Section~\ref{odd case}, we consider the case of odd characteristic.
Instead of dihedral codes,
 the consta-dihedral codes are considered; and 
the questions considered above are solved in a similar way.

\section{Preliminaries}\label{preliminaries}

Let $F$ be a finite field with $|F|=q$ as before, 
and $G$ be a finite group of order~$n$. 
Let $FG=\big\{\sum_{x\in G}a_x x\,\big|\, a_x\in F\big\}$ be the group algebra
(with the multiplication induced by the group multiplication of $G$). 
Any left ideal $C$ of $FG$, denoted by $C\le FG$, is called an {\em $FG$-code}.

The map $x\mapsto x^{-1}$ for $x\in G$ 
is an anti-automorphism of the group $G$, 
where~$x^{-1}$ denotes the inverse of $x$.  
We have an anti-automorphism of the algebra $FG$: 
\begin{equation}\label{bar map}\textstyle
 FG\longrightarrow FG,~~ \sum_{x\in G}a_x x
 \longmapsto \sum_{x\in G}a_x x^{-1}.
\end{equation}
We denote $\sum_{x\in G}a_x x^{-1}
 =\overline{\sum_{x\in G}a_x x}$,   
and call Eq.\eqref{bar map} the ``bar'' map of $FG$ for convenience. 
So, $\overline{\overline a}=a$, $\overline{a+b}=\overline a+\overline b$, 
$\overline{ab}=\overline b\,\overline a$
for $a,b\in FG$.
Note that the ``bar'' map is an automorphism of $FG$ 
once $G$ is abelian. 

The following is a linear form of $FG$:
$$\textstyle
\sigma:~ FG\longrightarrow F,~~ 
\sum_{x\in G}a_x x\longmapsto a_{1_G} ~~~
(\mbox{$1_G$ is the identity of $G$}).
$$

\begin{lemma}[{\cite[Lemma II.4]{FL20}}]\label{sigma bar}
{\bf(1)} $\sigma(ab)=\sigma(ba)$, $\forall~a,b\in FG$.

{\bf(2)} $\langle a, b\rangle=\sigma(a\overline b)=\sigma(\overline a b)$, 
          $\forall~a,b\in FG$.

{\bf(3)} $\langle d\,a,\,b\rangle=\langle a,\,\overline d\, b\rangle$, 
          $\forall~a, b,d\in FG$.

{\bf(4)} If $C$ is an $FG$-code, then so is $C^\bot$.

{\bf(5)} For $FG$-codes $C$ and $D$, 
 $\langle C,D\rangle=0$ if and only if $C\overline D=0$, 
where $C\overline D:=\big\{c\,\overline d\,\big|\,c\in C,d\in D\big\}$.
\end{lemma}

Assume that $H=\langle x\rangle=\{1,x,\cdots,x^{n-1}\}$ 
is a cyclic group of order $n$, and assume that $\gcd(n,q)=1$ as before. 
The group algebra
$FH=\big\{a=a(x)=\sum_{i=0}^{n-1}a_ix^i\;|\; a_i\in F\big\}$, and
$$
 F[X]\big/\langle X^n-1\rangle \cong FH, ~
 a(X)\mapsto a(x).
$$
Thus $FH$-codes are cyclic codes.  
Each cyclic code $C\le F[X]\big/\langle X^n-1\rangle$ 
is determined by generating polynomial $g(X)$,
and by check polynomial $h(X)$ also, where $h(X)g(X)=X^n-1$.
And,  $C\le FH$ is determined by an idempotent $e$ (i.e., $e^2=e$)
as follows: $c\in C$ if and only if $ec=c$; 
and also determined by the complementary idempotent $e'$ 
(satisfying that $ee'=0$ and $e+e'=1$):  $c\in C$ if and only if $e'c=0$.
Further, since $\gcd(n,q)=1$, 
$FH$ is semisimple and is uniquely decomposed 
into a direct sum of irreducible submodules (irreducible ideals) 
$FH e_i$ (the corresponding idempotents $e_i$ are said to be {\em primitive}):
\begin{equation}\label{FH=}
  FH=FHe_0\oplus FHe_1\oplus\cdots\oplus FH e_s. 
\end{equation}
Note that the identity $1=1_H$ of $H$ is also the identity of the algebra $FH$. 
The set $E=\big\{ e_0=\frac{1}{n}\sum_{i=0}^{n-1}x^i,~ e_1,~\cdots,~e_s\big\}$
of all primitive idempotents of $FH$ satisfies that
\begin{equation}\label{idempotent}
\textstyle
 1=e_0+e_1+\cdots+e_s; 
\qquad e_ie_j=\begin{cases} e_i, &i=j;\\ 0, & i\ne j. \end{cases}
\end{equation}
Please refer to \cite[Chapter 4]{HP} for details.
For any $a\in FH$,  $a=a1=\sum_{i=0}^sae_i$, where
$ae_i\in FHe_i$ is called the $e_i$-component of $a$.
Then the submodules of $FH$ can be characterized by the subsets of $E$
as follows.

\begin{lemma}\label{C e_C}
For any $C\le FH$, set $E_C=\{ \epsilon\in E\,|\, C\epsilon \ne 0\}$ 
and $e_C=\sum_{\epsilon\in E_C}\epsilon$. 
Then $C=\bigoplus_{\epsilon\in E_C}FH\epsilon=FHe_C$, 
which is a commutative ring with identity $e_C$.
\end{lemma}

\begin{proof}
If $\epsilon\in E_C$, then $0\ne C\epsilon \subseteq FH\epsilon$;
since $FH\epsilon$ is a simple ideal, we have $C\epsilon= FH\epsilon$.
Thus $FHe_C\subseteq C$. 
For any $c\in C$, if $\epsilon\in E-E_C$ then $c\epsilon=0$; 
so $c=\sum_{i=0}^s ce_i=\sum_{\epsilon \in E_C}c\epsilon \in FH e_C$. 
Thus $C\subseteq FHe_C$. We get $C= FHe_C$.
And, for $c\in C$, $c=ae_C$ for some $a\in FH$; then
$ce_C=ae_Ce_C=ae_C=c$. So $C$ is a ring with identity $e_C$.
\end{proof}

Since the ``bar'' map in Eq.\eqref{bar map} is an automorphism of $FH$ of order $2$, it permutes the primitive idempotents; i.e.,
$\overline\epsilon\in E$ for any $\epsilon\in E$.

\begin{lemma}\label{C C'}
Let $C,C'\le FH$. Denote 
$\overline E_{C}=\{\overline\epsilon\,|\,\epsilon\in E_{C}\}$.
Then:

{\rm(1)} $E_{C\cap C'}=E_C\cap E_{C'}$,  and $e_{C\cap C'}\!=\!e_Ce_{C'}$.

{\rm(2)} $E_{C+C'}=E_C\cup E_{C'}$, and $e_{C+C'}=e_C+e_{C'}\!-\!e_Ce_{C'}$. 
 
{\rm(3)} $E_{\overline C}=\overline E_C$, 
  and $e_{\overline C}=\overline e_C=\sum_{\epsilon\in\overline E_C}\epsilon$.
\end{lemma}
\begin{proof}
(1). $C\cap C'=\Big(\bigoplus_{\epsilon\in E_{C}}FH\epsilon\Big)
  \cap\Big(\bigoplus_{\epsilon\in E_{C'}}FH\epsilon\Big)
 =\bigoplus_{\epsilon\in E_{C}\cap E_{C'}}FH\epsilon$.
So $E_{C\cap C'}=E_C\cap E_{C'}$. 
And, $e_Ce_{C'}=\Big(\sum_{\epsilon\in E_{C}}\epsilon\Big)
  \Big(\sum_{\epsilon\in E_{C'}}\epsilon\Big)
 =\sum_{\epsilon\in E_{C}\cap E_{C'}}\epsilon
 =\sum_{\epsilon\in E_{C\cap C'}}\epsilon=e_{C\cap C'}$.

(2).  Denote $E_0=E_{C\cap C'}$, $E_1=E_C-E_0$, $E_2=E_{C'}-E_0$;
and $f_i=\sum_{\epsilon\in E_i}\epsilon$ for $i=0,1,2$.
Then 
\begin{align*}
C+C'&\textstyle
 =\Big(\bigoplus_{\epsilon\in E_1}FH\epsilon\Big)
  +\Big(\bigoplus_{\epsilon\in E_0}FH\epsilon\Big) 
 +\Big(\bigoplus_{\epsilon\in E_2}FH\epsilon\Big)
 +\Big(\bigoplus_{\epsilon\in E_0}FH\epsilon\Big)\\
&\textstyle
  =\bigoplus_{\epsilon\in E_0\cup E_1\cup E_2}FH\epsilon
  =\bigoplus_{\epsilon\in E_C\cup E_{C'}}FH\epsilon.
\end{align*}
That is, $E_{C+C'}=E_C\cup E_{C'}$. And,
\begin{align*}
&e_C+e_{C'}-e_Ce_{C'}=(f_1+f_0)+(f_2+f_0)-(f_1+f_0)(f_2+f_0)\\
&\textstyle
=f_1+f_2+2f_0 -f_0=f_1+f_2+f_0
=\sum_{\epsilon\in E_C\cup E_{C'}}\epsilon.
=e_{C+C'}.
\end{align*}

(3). $\overline C=\overline{\bigoplus_{\epsilon\in E_C}FH\epsilon}
 =\bigoplus_{\epsilon\in E_C}FH\overline\epsilon
 =\bigoplus_{\epsilon\in\overline E_C}FH\epsilon$.
\end{proof}

\begin{remark}\label{units of ring}\rm
(1) 
Any ring $R$ in this paper has identity. 
By $R^\times$ we denote the multiplicative group
consisting of all units (invertible elements) of~$R$.
It is known that 
the $R$-endomorphism ring ${\rm End}_{R}(R)$ 
of the regular $R$-module is just $R$ itself; because:
any $R$-homomorphism $\vph:R\to R$ corresponds to the element
$g:=\vph(1_R)$ such that $\vph(a)=\vph(a1_R)=a\vph(1_R)=ag$, $\forall$ $a\in R$.
And, $\vph$ is an $R$-isomorphism if and only if $g\in R^\times$.
An $R$-module $M$ is said to be {\em principal}
if $M$ generated by one element, i.e., $M=Rm$ for an $m\in M$. 

(2) Turn to the group algebra $FH$. 
For any $e_i\in E$, $FH e_i$ is a commutative simple ring, hence 
is a field; then $(FH e_i)^\times=FHe_i-\{0\}$.
Further, for $C=FHe_C=\bigoplus_{\epsilon\in E_C}FH\epsilon\le FH$, 
we have $C^\times=\prod_{\epsilon\in E_C}(FH\epsilon)^\times$.

(3) 
Note that the irreducible submodules $FHe_i$ for $i=0,1,\cdots,s$ 
are non-isomorphic each other; because: if $0\le i\ne j\le s$, 
then $e_iFHe_j=0$ but $e_iFHe_i=FHe_i\ne 0$, so $FHe_i\not\cong FHe_j$. 
Let $C,C'\le FH$ and $\vph:C\to C'$
be an $FH$-isomorphism. Then $E_{C'}=E_{C}$ hence $C'=C$
($C=FH e_C$ as in Lemma~\ref{C e_C}); 
and, setting $g=\vph(e_C)\in C^\times$, 
we have that 
\begin{equation}\label{auto of C}
\vph(c)=cg, \qquad \forall~c\in C. 
\end{equation}
Obviously, for any $g'\in FH$ such that $g'e_C =g$, 
we can get $\vph(c)=cg'$ $\forall$ $c\in C$. But the $g=\vph(e_C)$  
is the unique one who belongs to $C$ and makes Eq.\eqref{auto of C} held.
On the other hand, for any $g\in C^\times$, mapping $c\in C$
to $cg\in C$ is obviously an $FH$-automorphism of $C$;
consequently, $C=FH g$ (thus, any $FH$-code is principal).

(4) Similarly to $E_C$ in Lemma \ref{C e_C}, for $g\in FH$ we
set $e_g=\sum_{\epsilon\in E_g}\epsilon$
where $E_g=\{\epsilon\in E\,|\,g\epsilon\ne 0\}$.
Then $g=\sum_{\epsilon\in E_g}g\epsilon$  
and $FHg=\bigoplus_{\epsilon\in E_g}FH\epsilon=FHe_g$.
\end{remark}

\begin{lemma}\label{g g'} Let $g,g'\in FH$. Then 

{\rm(1)} $FH g\cap FH g'=FH (gg')$. In particular, 
$FH g\cap FH g'=0$ $\iff$ $gg'=0$.

{\rm(2)} 
If $gg'=0$, then $FHg+FH g'=FH g\oplus FH g'=FH\!\cdot\!(g+g')$.
\end{lemma}
\begin{proof}
(1). $FH g\cap FH g'=\big(\bigoplus_{\epsilon\in E_g}FH\epsilon\big)\cap
\big(\bigoplus_{\epsilon\in E_{g'}}FH\epsilon\big)
=\bigoplus_{\epsilon\in E_g\cap E_{g'}}FH\epsilon$.
And $E_g\cap E_{g'}=\{\epsilon\in E\,|\,g\epsilon\ne 0,g'\epsilon\ne 0\}$.
Because $FH\epsilon$ for $\epsilon\in E$ is a field, 
both $g\epsilon\ne 0$ and $g'\epsilon\ne 0$ if and only if $gg'\epsilon=(g\epsilon)(g'\epsilon)\ne 0$. So $E_g\cap E_{g'}=E_{gg'}$. Thus
$FH g\cap FH g'=\bigoplus_{\epsilon\in E_{gg'}}FH\epsilon=FH(gg')$.

(2). Since $gg'=0$, by (1), $FH g\cap FH g'=0$ and 
$E_g\cap E_{g'}=E_{gg'}=\emptyset$.
So, $FHg+FH g'=FH g\oplus FH g'$. 
And $E_{g+g'}=E_g\cup E_{g'}$. Thus $FHg+FH g'=FH(g+g')$.
\end{proof}

The set $E$ in Eq.\eqref{idempotent} of all primitive idempotents of $FH$
is invariant under the ``bar'' map in Eq.\eqref{bar map}, so
it is a disjoint union of $E'$ and $E''$ as follows:  
\begin{align}\label{E' E''}
E=E'\cup E'',\qquad 
E'=\{\epsilon\in E\,|\,\overline\epsilon=\epsilon\},\quad
E''=\{\epsilon\in E\,|\,\overline\epsilon\ne\epsilon\}.
\end{align}
The cardinality $|E''|\ge 0$, but $|E'|\ge 1$ because $e_0\in E'$. 

$\bullet$~ $|E'|=1$ if and only if $\overline e_i\ne e_i$
for $i=1,\cdots,s$.

$\bullet$~ $|E''|=0$ if and only if $\overline e_i=e_i$
for $i=0,1,\cdots,s$.

\begin{remark}\label{n and E'}\rm
Let ${\Bbb Z}_n$ be the integer residue ring modulo $n$, 
and ${\Bbb Z}_n^\times$ be the multiplicative unit group of ${\Bbb Z}_n$. 
Then $q\in{\Bbb Z}_n^\times$ (since $\gcd(n,q)=1$).
In the multiplicative group~${\Bbb Z}_n^\times$,
${\rm ord}_{{\Bbb Z}_n^\times}(q)$ denotes the order of $q$,
and $\big\langle q\big\rangle_{{\Bbb Z}_n^\times}$
denotes the cyclic subgroup generated by $q$.
The following two facts are well-known. 
 
{\bf(1)}\,{\rm(\cite[Theorem 6]{AKS})}\, 
{\it $E'=\{e_0\}$ if and only if
${\rm ord}_{{\Bbb Z}_n^\times}(q)$ is odd.}

{\bf(2)}\,{\rm(\cite[Theorem 1]{KR})}\,
{\it $E'=E$ (i.e., $E''=\emptyset$) if and only if
$-1\in\big\langle q\big\rangle_{{\Bbb Z}_n^\times}$.}

\noindent
Note that (please see \cite[Corollary II.8]{FL20}), 
there are infinitely many integers~$n$ satisfying the above~(1), 
and also infinitely many integers $n$ satisfying the above~(2).
Further, there are also infinitely many integers $n$ 
who meet none of the the above two cases. 
\end{remark}

\begin{lemma}\label{l C^bot}
For any $e_i\in E$, the orthogonal ideal
$(FH e_i)^\bot=\bigoplus_{\epsilon\in E-\{\overline e_i\}}FH\epsilon$.
In particular,

{\rm(1)} $FH e_i$ is LCD if $e_i\in E'$; 

{\rm(2)} $FH e_i$ is self-orthogonal if $e_i\in E''$.
\end{lemma}
\begin{proof}
By Lemma \ref{sigma bar}(5),
$\big\langle FH e_i,\,FH e_j\big\rangle=0$ if and only if 
$(FH e_i)(\overline{FH e_j})=FH e_i\overline e_j=0$, if and only if
$e_i\overline e_j=0$, if and only if $e_j\ne \overline e_i$. 
\end{proof}

\begin{corollary}\label{c C^bot}
For $C\le FH$, the orthogonal ideal
$C^\bot=\bigoplus_{\epsilon\in E-\overline E_{C}}FH\epsilon$.
\end{corollary}

\section{$2$-quasi-cyclic codes and dihedral codes}\label{2-quasi-cyclic}

In this section we characterize $2$-quasi-cyclic codes by 
the so-called Goursat Lemma (Theorem \ref{Goursat C}), 
and then (in Theorem \ref{2-quasi dihedral}) 
answer the question: what $2$-quasi-cyclic codes are dihedral codes?

\subsection{$2$-quasi-cyclic codes}
Let $H=\{1,x,\cdots,x^{n-1}\}$ be the cyclic group of order $n$, 
and $\gcd(n,q)=1$ as before.
Any $FH$-submodule $C$ of $(FH)^2=FH\times FH$, 
denoted by $C\le(FH)^2$, is called a 
a {\em quasi-cyclic code} over $F$ of {\em index} $2$ and {\em coindex} $n$, 
or  a {\em $2$-quasi-cyclic code} over $F$ of length $2n$. 
For $c\in C$ we write $c=(c,c')$ with $c,c'\in FH$. 
We write $c=c(x)=c_0+c_1x+\cdots+c_{n-1}x^{n-1}$ 
if we consider the polynomial version $FH\cong F[X]/\langle X^n-1\rangle$; 
and we write $c=(c_0, c_1,\cdots, c_{n-1})$ if we consider the word version.

\begin{remark}
\rm
Goursat Lemma is originally a group-theoretic result, e.g., see \cite[p25]{AB}.
It is extended in similar way to rings, modules etc.,
 e.g., see \cite{Lambek}. We state the module version of Goursat Lemma
as follows (cf. \cite{AB, Lambek}, or one can check it step by step).

\noindent{\bf Goursat Lemma}. 
{\it Let $R$ be a ring and $M_1\times M_2$ be the direct sum of 
 $R$-modules $M_1,M_2$. 
By $\rho_i: M_1\times M_2 \to M_i$, $i=1,2$,
we denote the projection from $M_1\times M_2$ to $M_i$. 

{\rm (1)} Let $C$ be a submodule of $M_1\times M_2$.
Denote 
\begin{equation}\label{C_i}
\begin{array}{l}
\widetilde C_1=\rho_1(C)
 =\{a_1\in M_1\,|\, (a_1,a_2)\in C~\mbox{for some~} a_2\in M_2\},
\\[2pt]
 C_1=\rho_1\big(C\cap(M_1\times 0)\big)
  =\{a_1\in M_1\,|\, (a_1,0)\in C\};
\\[2pt]
 \widetilde C_2=\rho_2(C)
 =\{a_2\in M_2\,|\, (a_1,a_2)\in C~\mbox{for some~} a_1\in M_1\},
\\[2pt]
 C_2=\rho_2\big(C\cap(0\times M_2)\big)=\{a_2\in M_2\,|\, (0,a_2)\in C\}.
\end{array}
\end{equation}
Then $C_i\le\widetilde C_i\le M_i$, $i=1,2$.
And, for any $c_1 + C_1\in \widetilde C_1/C_1$ 
there is a unique $c_2+C_2\in\widetilde C_2/C_2$ such that 
$(c_1,c_2)\in C$, and the map 
\begin{equation}\label{varphi}
 \vph: \widetilde C_1/C_1\to \widetilde C_2/C_2, ~
 c_1+C_1\mapsto c_2+C_2, \quad (\mbox{where $(c_1,c_2)\in C$,})
\end{equation}
is an $R$-isomorphism, and
\begin{equation}\label{structure}
 C=\big\{\,(c_1,c_2)\;\big|\; 
  c_i\in\widetilde C_i, i=1,2; \,\vph(c_1+C_1)=c_2+C_2\,\big\}.
\end{equation}

{\rm (2)} If $C_i\le\widetilde C_i\le M_i$ for $i=1,2$, 
and $\vph:  \widetilde C_1/C_1\to \widetilde C_2/C_2$ is an $R$-isomorphism, then 
the $C$ constructed in Eq.\eqref{structure} is an $R$-submodule of $M_1\times M_2$.
}
\end{remark}

Return to $2$-quasi-cyclic codes. 
Since $FH$ is semisimple, by Lemma \ref{C e_C}
we can get a refined version of Goursat Lemma for $C\le FH\times FH$.

\begin{theorem}\label{Goursat C}
If $C\le (FH)^2$,
then there are $C_1,C_2,C_{12}\le FH$ 
satisfying that $C_1\cap C_{12}=0=C_2\cap C_{12}$, 
and an element $g\in C_{12}^\times$ such that 
\begin{equation}\label{structure C}
 C = (C_1\times C_2)\oplus\widehat C_{12},~~ \mbox{where}~~
\widehat C_{12}=
 \big\{\big(c_{12},\,c_{12}g\big)\,\big|\,c_{12}\in C_{12}\big\}\cong C_{12}.
\end{equation}

Conversely, if there are $C_1, C_2, C_{12}\le FH$ 
with $C_1\cap C_{12}=0=C_2\cap C_{12}$ and a
$g\in C_{12}^\times$, then the $C$ constructed in Eq.\eqref{structure C}
is an $FH$-submodule of $(FH)^2$. 
\end{theorem}

\begin{proof}
Assume that $C\le (FH)^2$. 
Let $C_1,C_2,\widetilde C_1,\widetilde C_2$ be as in Eq.\eqref{C_i}.
Since $C_1\le\widetilde C_1\le FH$, by Lemma \ref{C e_C} and Lemma \ref{C C'},
we can assume that 
$E_{C_1}=\{e_{i_1},\cdots,e_{i_k}\}\subseteq E$,
i.e., $C_{1}=\bigoplus_{\alpha=1}^{k}FHe_{i_\alpha}$; and that
$$E_{\widetilde C_1}=\{e_{i_1},\cdots,e_{i_k},\,e_{j_1},\cdots,e_{j_{h}}\}\subseteq E,
\mbox{~~~with}~~ \{j_1,\cdots,j_{h}\}\subseteq E-E_{C_1};
$$
Then $\widetilde C_1=C_1\oplus C_{12}$, 
where $C_{12}=\bigoplus_{\alpha=1}^{h}FHe_{j_\alpha}$, i.e., 
$E_{C_{12}}=\{j_1,\cdots,j_{h}\}$.
Similarly, we have a $C_{12}'\le FH$ such that $\widetilde C_2=C_2\oplus C_{12}'$.
Then $\widetilde C_1/C_1\cong C_{12}$ and $\widetilde C_2/C_2\cong C'_{12}$; so
the isomorphism $\vph$ in Eq.\eqref{varphi} induces an
isomorphism 
\begin{equation}\label{vph'}
\vph': C_{12}\to C_{12}'~~\mbox{ such that }~~ 
\vph\big(c+C_1\big)=\vph'(c)+C_2,~\forall~c\in C_{12}.
\end{equation}
So $C_{12}'=C_{12}$ and there is an element $g\in C_{12}^\times$
such that $\vph'(c)=cg$ for all $c\in C_{12}$; 
see Remark \ref{units of ring}(3).
In conclusion,
we have a $C_{12}\le FH$ such that
\begin{equation*}
 \widetilde C_1 =C_1\oplus C_{12}\le FH,\quad 
 \widetilde C_2 =C_2\oplus C_{12}\le FH;
\end{equation*} 
and a $g\in C_{12}^\times$ such that
(note that $(c_1\!+\!c_{12},\,c_2\!+\!c_{12}g)=(c_1,c_2)+(c_{12},c_{12}g)$):
\begin{equation}\label{more structure C}
 C\!=\!\big\{(c_1\!+\!c_{12},\;c_2\!+\!c_{12}g)\,\big|\,
 c_1\!\in\! C_1, c_2\!\in\! C_2,  c_{12}\!\in\! C_{12}\big\}
\!=\!(C_1\times C_2)\oplus \widehat C_{12}.
\end{equation}

By Goursat Lemma, the converse part is obviously true.
\end{proof}

\begin{corollary}\label{dim C}
Let $C$ be as in Eq.\eqref{structure C}. 
Then $C=(C_1\times 0)\oplus(0\times C_2)\oplus\widehat C_{12}$, 
the $FH$-composition factors of $C$ consist of the composition factors of 
$C_1$, $C_2$ and $C_{12}$; in particular, 
$\dim_F C=\dim_F C_1+\dim_F C_2+\dim_F C_{12}$.
\end{corollary}

Recall that $C\le(FH)^2$ is said to be principal 
if $C$ is generated by one element, see Remark~\ref{units of ring}(1).

\begin{corollary}\label{generators of C}
{\rm(1)}
Any $2$-quasi-cyclic codes over $F$ is generated by two elements.

{\rm(2)} A $2$-quasi-cyclic code $C\le(FH)^2$ is principal 
if and only if $C$ has no repeated $FH$-composition factors,
i.e., $E_{C_1}\!\cap E_{C_2}=E_{C_1}\cap E_{C_{12}}
 =E_{C_2}\cap E_{C_{12}}=\emptyset$.
\end{corollary}

\begin{proof}
(1). Take the notation in Eq.\eqref{structure C}. 
By Remark~\ref{units of ring}(3), there are 
$g_1\in C_1$, $g_{12}\in C_{12}$ and $g_2\in C_2$ 
such that $C_1=FH g_1$, $C_{12}=FH g_{12}$ and $C_2=FH g_2$. 
Take $(g_1+g_{12},\, g_{12}g),(0,g_2) \in C$, and $e_{C_1}$ as in Lemma \ref{C e_C}. 
Since $C_1\cap C_{12}=0$,   
$$e_{C_1}(g_1+g_{12},\, g_{12}g)
  =\big(e_{C_1}g_1+e_{C_1}g_{12},\; e_{C_1}g_{12}g)=(g_1,0);
$$
and $FH(g_1,0)=C_1\times 0$. Similarly, 
$e_{C_{12}}(g_1+g_{12},\, g_{12}g)=(g_{12}, g_{12}g)$ and
$FH(g_{12},\,g_{12}g)=\widehat C_{12}$. 
By the stricture of $C$ in Eq.\eqref{structure C}
(or by Corollary \ref{dim C}), 
\begin{align}\label{2 generators}
C=FH\!\cdot\!(g_1+g_{12},\; g_{12}g)+FH\!\cdot\!(0, g_2). 
\end{align}

(2).
If $C=FH\!\cdot\!(a,b)$ is generated by $(a,b)\in(FH)^2$, 
then $FH\to C$, $d\mapsto d(a,b)$, is a surjective homomorphism.
The regular module $FH$ has no repeated composition factors
(cf. Remark~\ref{units of ring}(3)), 
hence $C$ has no repeated composition factors.

Next assume that $C$ has no repeated composition factors.
Let $g_1,g_2,g_{12},g$ be as above.
By Corollary \ref{dim C}, 
$FH g_1$, $FH g_2$ and $FH g_{12}$ have no common $FH$-composition factors.
Then $e_{C_1}(g_1+g_{12},\, g_2+g_{12}g)=(g_1,0)$, 
$e_{C_2}(g_1+g_{12},\, g_2+g_{12}g)=(0,g_2)$,
$e_{C_{12}}(g_1+g_{12},\, g_2+g_{12}g)=(g_{12},g_{12}g)$.
Thus, $C=FH f$ for $f=(g_1+g_{12},\, g_2+g_{12}g)$.
\end{proof}

\begin{theorem}\label{generator matrix}
Let $C\le(FH)^2$. 
The following three are equivalent to each other:

{\rm(1)} $C$ is a double circulant code of length $2n$, i.e., 
as a linear code, $C$ has a double circulant matrix $(I~A)$ of size $n\times 2n$
as a generator matrix. 

{\rm(2)} $C=FH\!\cdot\!(1,a)$ for an element $(1,a)\in (FH)^2$.

{\rm(3)}  $\widetilde C_1=C_1+C_{12}=FH$ and $C_2=0$, 
where $C_1, C_2, C_{12}$ are as in Eq.\eqref{structure C}.
\end{theorem}

\begin{proof}
(1) $\Leftrightarrow$ (2). \cite[Lemma 3.7]{FL21}.

(2) $\Rightarrow$ (3). Since $C=\{(a 1, a g)\,|\,a\in FH\}$, 
by the definition of $\widetilde C_1$ in Eq.\eqref{C_i}, 
$\widetilde C_1=FH\!\cdot\! 1=FH$. By Corollary \ref{generators of C}(2),
$C_2$ and $\widetilde C_1=FH$ have no composition factors in common, hence $C_2=0$.

(3) $\Rightarrow$ (2). Since $C_1+C_{12}=\widetilde C_1=FH$, 
in Eq.\eqref{2 generators} we can take $g_1+g_{12}=1$. 
Since $C_2=0$, we get that $g_2=0$, and by Eq.\eqref{2 generators}, 
$C=FH\!\cdot\!(1,\,g_{12}g)$.
\end{proof}

\subsection{Dihedral group codes}

Let $FH$ as before.
Let $T=\{1,y\}$ be a cyclic group of order $2$,
and $G=H\rtimes T$ be the semidirect product with the relation $yxy^{-1}=x^{-1}$,
i.e., $G$ is the the dihedral group of order $2n$.
Then  
$G=\{1,x,\cdots,x^{n-1},~ y, xy,\cdots, x^{n-1}y\}$
 consists of $2n$ elements.
The group algebra $FG$ is an $F$-space with basis $G$. Hence
\begin{equation}\label{E FG=}
 FG=\big\{a(x)+a'(x)y\:\big|\; a(x),a'(x)\in FH \big\}.
\end{equation}
Any left ideal $C$ (i.e., any $FG$-submodule of the left regular module $FG$)
is denoted by $C\le FG$, and called a {\em dihedral code} over $F$ of length $2n$.

\begin{remark}\label{r FG=}\rm
 $FG$ has a subalgebra $FH$ such that
$FG=FH\oplus FH y$ 
with left multiplication by $y$ as follows: for $a(x)+a'(x)y\in FH\oplus FHy$,
\begin{equation}\label{left by y}
  y\big(a(x)+a'(x)y\big)=ya(x)+ya'(x)y=\overline{a'}(x)+\overline a(x)y.
\end{equation}
If we restrict the left regular $FG$-module $FG$ to an $FH$-module, then
we have an $FH$-isomorphism:
\begin{align}\label{dihedral to 2-cyclic}
 FG~ \stackrel{\cong}{\longrightarrow}~ FH\times FH, ~~ 
a(x)+a'(x)y~\longmapsto~ \big(a(x),\,a'(x)\big).
\end{align}
Thus, any dihedral code $C\le FG$ is restricted, 
in the natural way of~Eq.\eqref{dihedral to 2-cyclic},
to a $2$-quasi-cyclic code $C\le(FH)^2$; or shortly speaking,
dihedral codes are $2$-quasi-cyclic codes. 
The converse is not true.
Thus a natural question follows:

$\bullet$~ 
{\it What $2$-quasi-cyclic codes are dihedral codes? }

\noindent
And it is a special concern (e.g., \cite{AOS}, \cite{AGOSS}, \cite{MW},\cite{Musa})
that:

$\bullet$~ {\it What self-dual $2$-quasi-cyclic codes 
are dihedral codes?  }

\noindent
We answer the first question in this section. The special concern is left to 
Section~\ref{self-dual dihedral}, because we need more information on self-dual
$2$-quasi-cyclic codes to solve it.
\end{remark}

\begin{theorem}\label{2-quasi dihedral}
Let $C=(C_1\times C_2)\oplus\widehat C_{12}\le (FH)^2$ 
be a $2$-quasi-cyclic code, 
where $C_1,C_2,C_{12}\le FH$, $g\in C_{12}^\times$ 
and $\widehat C_{12}\le(FH)^2$ are as in Eq.\eqref{structure C}.
The following two are equivalent.

{\rm(1)} $C$ is a dihedral code.

{\rm(2)} $\overline C_1=C_2$, $\overline C_{12}=C_{12}$
and $\overline g g=e_{C_{12}}$, 
where $e_{C_{12}}$ is defined in Lemma~\ref{C e_C}.
\end{theorem}

\begin{proof} (1)$\Rightarrow$(2). 
Assume that $C\le FG$. 
For any $c(x)+c'(x)y \in C$, we have (cf. Eq.\eqref{left by y}):
$$
y\big(c(x)+c'(x)y\big) =\overline{c'}(x)+\overline{c}(x)y \in C.
$$
Thus $y(C_1+ 0y)=0+\overline C_1y\subseteq C$, 
hence $\overline C_1\subseteq C_2$. 
Similarly, $\overline C_2\subseteq C_1$, i.e., 
$C_2\subseteq\overline C_1$. 
Thus $\overline C_1=C_2$, equivalently, $\overline E_{C_1}=E_{C_2}$. 
Next, for any $c(x)\in C_{12}$,  
\begin{equation}\label{dihdral 1}
 c(x)+(c(x)g(x))y\in C ~~\implies~~
 \overline{(c(x)g(x))}+\overline{c(x)}y\in C.
\end{equation} 
Take $\epsilon\in E_{C_{12}}$; we have
$\overline\epsilon (x)\overline g(x)+\overline\epsilon (x)y\in C$;
so $\overline\epsilon\in E_{C_{12}}$ 
(otherwise, $\overline\epsilon\in E_{C_2}$ hence 
$\epsilon\in\overline E_{C_2}= E_{C_1}$, 
a contradiction to that $E_{C_1}\cap E_{C_{12}}=\emptyset$).
So $\overline E_{C_{12}}=E_{C_{12}}$, hence 
$\overline e_{C_{12}}=e_{C_{12}}$, 
$\overline C_{12}=C_{12}$; and   
$\overline\epsilon (x)\overline g(x)+\overline\epsilon (x)y\in \widehat C_{12}$.
Thus 
$$\textstyle
 e_{C_{12}} (x)\overline g(x)+ e_{C_{12}} (x)y
=\sum_{\epsilon\in E_{C_{12}}}\!
\big(\overline\epsilon (x)\overline g(x)+\overline\epsilon (x)y\big) \in\widehat C_{12}.
$$
By the definition of $\widehat C_{12}$, 
$e_{C_{12}} (x)\overline g(x)\cdot g(x)=e_{C_{12}} (x)$.
Recalling that $e_{C_{12}}$ is the identity of $C_{12}$ and $g\in C_{12}^\times$,
we get $\overline{g}g=e_{C_{12}}$.

(2)$\Rightarrow$(1).  Let $C\le(FH)^2$ 
as in Theorem \ref{Goursat C}, and assume that 
$\overline C_1=C_2$, $\overline C_{12}=C_{12}$ and $\overline g g=e_{C_{12}}$.
Identify $C$ with the subset 
$$\{c(x)+c'(x)y\,|\,(c(x),c'(x))\in C\} \subseteq FG. $$
Then $C$ is obviously invariant by left $FH$-multiplication. It remains to prove
that $C$ is invariant by left multiplication by $y$. 
For $c(x)+c'(x)y\in C$,  by Eq.\eqref{structure C} we can set
$c(x)=c_1(x)+c_{12}(x)$ and $c'(x)=c_2(x)+c_{12}(x)g(x)$ 
with $c_1(x)\in C_1$, $c_2(x)\in C_2$ and $c_{12}(x)\in C_{12}$.
Then
$$
y\big(c(x)+c'(x)y\big)=\overline{c'(x)}+\overline{c(x)}\,y
=\overline{(c_2(x)+c_{12}(x)g(x))}+\overline{(c_1(x)+c_{12}(x))}\;y,
$$
i.e.,
\begin{equation}\label{dihedral 2}
y\big(c(x)+c'(x)y\big)
=\big(\overline c_2(x) +\overline c_{12}(x)\overline g(x)\big)
 + \big(\overline c_1(x)+\overline c_{12}(x)\big)\,y.
\end{equation}
Because $\overline C_1=C_2$, $\overline C_{12}=C_{12}$, we have
$$
\overline c_2(x)\in C_1, ~~~~~ \overline c_1(x)\in C_2; ~~~~~
\overline c_{12}(x)\overline g(x), ~
\overline c_{12}(x)\in C_{12}. 
$$
Further, by the assumption that $g\overline g=e_{C_{12}}$, we get
\begin{equation}\label{dihedral 3}
\overline c_{12}(x)\overline g(x)\cdot g(x)=
\overline c_{12}(x) e_{C_{12}}(x)=\overline c_{12}(x).
\end{equation}
By Theorem \ref{Goursat C} and cf. Eq.\eqref{more structure C},
we conclude that
$$
y\big(c(x)+c'(x)y\big)\in C,\quad \forall~ c(x)+c'(x)y\in C.
$$
we are done.
\end{proof}

\section{Self-dual $2$-quasi-cyclic codes}\label{self-dual 2-quasi-cyclic}
\def\cir{\mbox{\scriptsize$n$-cir}}

Keep the notation in Section \ref{2-quasi-cyclic}.
With the structure in Eq.\eqref{structure C}, 
in this section we characterize the self-dual $2$-quasi-cyclic codes.

Let ${\rm M}_{\cir}(F)$ be the set of the $n\times n$ circulant matrices over $F$.
Any $a(x)=a_0+a_1x+\cdots+a_{n-1}x^{n-1}\in FH$ 
determines a circulant matrix in ${\rm M}_{\cir}(F)$, whose first row is 
$(a_0,a_1,\cdots,a_{n-1})$, and each next row is obtained by
right shift the previous row.  
Let $P$ be the circulant matrix determined by $x$, i.e.,
\begin{equation}\label{P}
P=\begin{pmatrix}
0 & 1 & 0 & \cdots & 0\\
0 & 0 & 1 & \cdots & 0\\
\vdots & \vdots & \ddots & \ddots & \vdots\\
0 & 0 & 0 & \ddots & 1\\
1 & 0 & 0 & \cdots & 0
\end{pmatrix}.
\end{equation}
Then the circulant matrix determined by $a(x)=a_0+a_1x+\cdots+a_{n-1}x^{n-1}$ is
$$
  a(P)=a_0 I+a_1P+\cdots+a_{n-1}P^{n-1},
$$ 
where $I$ is the identity matrix. 
It is well-known (e.g., cf. \cite{G06}, and also easy to check) that 
\begin{equation}\label{FH cir}
FH~\longrightarrow~{\rm M}_{\cir}(F),~~
a(x)~\longmapsto~a(P),
\end{equation}
is an $F$-algebra isomorphism.

For any matrix $A$, by $A^T$ we denote the transpose of $A$.

\begin{lemma}\label{a^T=bar a}
For $a(x)=a_0+a_1x+\cdots+a_{n-1}x^{n-1}\in FH$, 
$$
\overline a(P)=a(P)^T.
$$
\end{lemma}

\begin{proof} By definition, $\overline a(x)=a_0+a_1x^{-1}+\cdots+a_{n-1}x^{-(n-1)}$. 
The inverse of $P$ in Eq.\eqref{P} is as follows:
\begin{equation}
P^{-1}=\begin{pmatrix}
0 & 0 & \cdots & 0 & 1\\
1 & 0 & \cdots & 0 & 0\\
0 & 1 & \ddots & 0 & 0\\
\vdots & \vdots & \ddots & \ddots & \vdots  \\
0 & 0 & \cdots & 1 & 0
\end{pmatrix}=P^T.
\end{equation}
Thus 
$\overline a(P)=a_0 I+a_1P^{-1}+\cdots+a_{n-1}P^{-(n-1)}$, 
whose first column is the transpose of $(a_0,a_1,\cdots,a_{n-1})$, 
and each next column is obtained by down shift the previous column. 
That is, $\overline a(P)=a(P)^T$.
\end{proof}

To distinguish the inner products on $FH$ and on $(FH)^2$, 
we denote them by $\langle -,-\rangle_{F\!H}$ and 
$\langle -,-\rangle_{(F\!H)^2}$
respectively. 

\begin{theorem}\label{self-dual C}
Let $C=(C_1\times C_2)\oplus\widehat C_{12}\le FH\times FH$
with $C_1,C_2,C_{12}\le FH$ and $g\in C_{12}^\times$ 
as in Theorem \ref{Goursat C}. 
Then the following {\rm(1)} and {\rm(2)} are equivalent to each other:

{\rm(1)} $C$ is self-dual.

{\rm(2)} The following three hold:

\quad {\rm(2.1)} $\dim_F C=n$, $\dim_F C_1=\dim_F C_2$; 

\quad {\rm(2.2)} 
$\big\langle C_1, C_1\big\rangle_{F\!H}=\big\langle C_2, C_2\big\rangle_{F\!H}
=\big\langle C_1, C_{12}\big\rangle_{F\!H}
=\big\langle C_2, C_{12}\big\rangle_{F\!H}=0$;

\quad {\rm(2.3)} $\overline C_{12}=C_{12}$ and $g\overline g=-e_{C_{12}}$, 
where $e_{C_{12}}$ is defined in Lemma \ref{C e_C}.
\end{theorem}

\begin{proof}
(1)$\Rightarrow$(2). It is trivial that $\dim_F C=n$. 
By the self-duality of $C$,  
any two submodules (not necessarily different) of $C$ are orthogonal. 
By Corollary~\ref{dim C},
$C=(C_1\times 0)\oplus(0\times C_2)\oplus\widehat C_{12}$. 
Thus\vskip-2mm
\begin{equation}\label{inner FH FH^2}
\begin{array}{l}
\big\langle C_1, C_1\big\rangle_{F\!H}
 =\big\langle C_1\times 0, C_1\times 0\big\rangle_{(F\!H)^2}=0;
\\[3pt]
\big\langle C_2, C_2\big\rangle_{F\!H}
 =\big\langle 0\times C_2, 0\times C_2\big\rangle_{(F\!H)^2}=0;
\\[3pt]
\big\langle C_1, C_{12}\big\rangle_{F\!H}
 =\big\langle C_1\times 0, \widehat C_{12}\big\rangle_{(F\!H)^2}=0;
\\[3pt]
\big\langle C_2, C_{12}\big\rangle_{F\!H}
 =\big\langle 0\times C_2, \widehat C_{12}\big\rangle_{(F\!H)^2}=0.
\end{array}
\end{equation} 
Without lose of generality, we can assume that $\dim_F C_1\ge \dim_F C_2$. 
Note that $\big\langle C_1, C_1\big\rangle_{F\!H}
=\big\langle C_1, C_{12}\big\rangle_{F\!H}=0$ and $C_1\cap C_{12}=0$.  
Using notation in Lemma~\ref{C e_C}, 
by Lemma \ref{l C^bot} and its corollary, we have 
$$
 E_{C_1}\cap\overline E_{C_1}=\emptyset,~~
 E_{C_1}\cup\overline E_{C_1}\subseteq E'', ~~ 
 E_{C_{12}}\subseteq E-(E_{C_1}\cup\overline E_{C_1}).
$$
By Lemma \ref{C C'}, 
we see that the submodule $\overline C_1=FH\overline e_{C_1}$ 
corresponds to the subset $\overline E_{C_1}\subseteq E$; hence
$C_1\cap\overline C_1=0$; cf. Lemma \ref{C C'}. Denote 
$$\textstyle
  e'=\sum_{\epsilon\in E-(E_{C_1}\cup\overline E_{C_1})}\epsilon.
$$
Then $\dim_F C_{12}=\dim_F FHe_{C_{12}}\le\dim_F FHe'$. 
By Corollary \ref{dim C}, we get
\begin{align*}
n=\dim_F C &=\dim_F C_1+\dim_F C_2+\dim_F C_{12}\\
 & \le \dim_F C_1 + \dim_F \overline C_1 +\dim_F FH e' =n.
\end{align*}
Therefore, both
``$\dim_F C_2\le\dim_F C_1$'' and 
``$\dim_F FH e_{C_{12}}\le\dim_F FHe'$'' have to be equalities, i.e., 
$$
 \dim C_2=\dim C_1 ~~~ \mbox{and}~~~
 E_{C_{12}}= E-(E_{C_1}\cup\overline E_{C_1}).
$$
The latter one implies that $\overline E_{C_{12}}=E_{C_{12}}$, hence
$\overline e_{C_{12}}=e_{C_{12}}$ and $\overline C_{12}=C_{12}$.

Finally, by Eq.\eqref{structure C}, 
$(e_{C_{12}},g)$ is a generator of the $FH$-submodule $\widehat C_{12}$ 
of $(FH)^2$, i.e., $\widehat C_{12}=FH\!\cdot\!(e_{C_{12}},g)$.
Let $e_{C_{12}}(P)$, $g(P)$ be the circulant matrices
determined by the $e_{C_{12}}(x)$, $g(x)\in FH$ respectively, 
cf.~Eq.\eqref{FH cir}.
Because $(e_{C_{12}},\,g)$ generates $\widehat C_{12}$, 
 the rows of the double circulant matrix $\big(e_{C_{12}}(P)\,~g(P)\big)$ 
 linearly generate $\widehat C_{12}$.
Since $\big\langle \widehat C_{12},\,\widehat C_{12}\big\rangle_{(F\!H)^2}=0$,
we get that
\begin{equation}\label{AA^T}
e_{C_{12}}(P)e_{C_{12}}(P)^T+g(P)g(P)^T=
\big(e_{C_{12}}\!(P)\,~ g(P)\big)\big(e_{C_{12}}\!(P)\,~ g(P)\big)^T=0.
\end{equation}
By Lemma \ref{a^T=bar a}
$$
e_{C_{12}}(P)\overline e_{C_{12}}(P)+g(P)\overline g(P)=0.
$$ 
By the isomorphism Eq.\eqref{FH cir}, we get
\begin{equation}\label{e bar e}
(e_{C_{12}}\overline e_{C_{12}} + g\overline g)(P)=0,
\end{equation}
and hence $e_{C_{12}}\overline e_{C_{12}} + g\overline g=0$, i.e., 
$g\overline g=-e_{C_{12}}$, 
because $\overline e_{C_{12}}=e_{C_{12}}$.

(2)$\Rightarrow$(1). 
The above arguments are in fact invertible.
From the assumption~(2.2) and cf.~Eq.\eqref{inner FH FH^2}, we can get
\begin{align*}
 &\big\langle C_1\times 0,\,C_1\times 0 \big\rangle_{(F\!H)^2}=0,\quad
 \big\langle C_1\times 0,\, \widehat C_{12} \big\rangle_{(F\!H)^2}=0;\\
 &\big\langle 0\times C_2,\, 0\times C_2 \big\rangle_{(F\!H)^2}=0, \quad
 \big\langle 0\times C_2,\,\widehat C_{12} \big\rangle_{(F\!H)^2}=0.
\end{align*}
By the assumption~(2.3), we can get Eq.\eqref{e bar e}, 
and then Eq.\eqref{AA^T}, and therefore, 
$\big\langle \widehat C_{12},\, \widehat C_{12} \big\rangle_{(F\!H)^2}=0$.
It is trivial that 
$\big\langle C_1\times 0,\,0\times C_2 \big\rangle_{(F\!H)^2}=0$.
Since $C=(C_1\times 0)\oplus(0\times C_2)\oplus\widehat C_{12}$, 
see Corollary \ref{dim C}, we see that $C$ is self-orthogonal.
Finally, by the assumption~(2.1), $C$ is self-dual, i.e., (1) holds.
\end{proof}

\section{Self-dual dihedral codes}\label{self-dual dihedral}
In this section we assume that the characteristic ${\rm char} F=2$, and
show complete solutions of Question 3 and Question 4 raised in Introduction.
As for the case of odd ${\rm char} F$, 
we'll answer the questions in the next section.

Note that we have assumed that $\gcd(n,q)=1$, so $n$ is odd in this section.

Let $H=\{1,x,\cdots,x^{n-1}\}$, $T=\{1,y\}$, 
$G=H\rtimes T$ with $yxy^{-1}=x^{-1}$ as before, 
and $FG$ be the dihedral group algebra as in Eq.\eqref{E FG=}.

\begin{theorem}\label{t char even}
Keep the assumption as above (specifically, ${\rm char} F=2$).
Let $C=(C_1\times C_2)\oplus \widehat C_{12}\le(FH)^2$ in Eq.\eqref{structure C}
be a self-dual $2$-quasi-cyclic code.
Then $C$ is a dihedral code if and only if $\overline C_1=C_2$.  
\end{theorem}

\begin{proof}
Since $C$ is self-dual, by Theorem \ref{self-dual C},  $C$ 
satisfies the three conditions in Theorem \ref{self-dual C}(2);
in particular,  
$\overline C_{12}=C_{12}$ and $g\overline g=-e_{C_{12}}$,  
      cf. Theorem \ref{self-dual C}(2.3).
Since ${\rm char} F=2$, $g\overline g=-e_{C_{12}}=e_{C_{12}}$.
By Theorem \ref{2-quasi dihedral}, 
the theorem holds at once.
\end{proof}

The result \cite[Theorem 2]{AOS} is extended as follows.

\begin{corollary}\label{principal self-dual}
Any principal self-dual $2$-quasi-cyclic code over~$F$ 
(with ${\rm char} F=2$) is dihedral.
In particular, any self-dual double circulant code over~$F$ 
(with ${\rm char} F=2$) is dihedral.
\end{corollary}
\begin{proof}
Let $C$ be a principal self-dual $2$-quasi-cyclic code over~$F$.
By Corollary~\ref{generators of C}(2), 
$E_{C_1}\cap E_{C_2}=E_{C_1}\cap E_{C_{12}}
 =E_{C_2}\cap E_{C_{12}}=\emptyset$. 
By Theorem~\ref{self-dual C}(2), $\dim_FC=n$, 
$\dim_F C_1=\dim_F C_2$,
$\langle C_1,C_1\rangle_{FH}=0$ and 
$\langle C_1,C_{12}\rangle_{FH}=0$.
Thus, $E_{C_{12}}\cup E_{C_1}\cup E_{C_2}=E$ (as $\dim_FC=n$),
and $E_{C_1}\cap \overline E_{C_1}
=E_{C_{12}}\cap \overline E_{C_1}=\emptyset$.
Then we have 
$E_{C_{12}}\cup E_{C_1}\cup \overline E_{C_1}
\subseteq E_{C_{12}}\cup E_{C_1}\cup E_{C_2}$, 
where the both sides are disjoint unions. So
$\overline E_{C_1}\subseteq E_{C_2}$, 
i.e., $\overline C_1\subseteq C_2$.
Hence $\overline C_1 = C_2$ (as $\dim_F C_1=\dim_F C_2$).
By Theorem~\ref{t char even}, $C$ is dihedral.

Any self-dual double circulant code over~$F$
 is principal (see Theorem~\ref{generator matrix}), hence is dihedral.
\end{proof}

We answer Question 4 in introduction: 
in what case any self-dual $2$-quasi-cyclic code is dihedral? 
We take the notation in Eq.\eqref{E' E''} and Remark \ref{n and E'}.
The expression $2^v\Vert m$ means that $2^v\!\mid\! m$ and $2^{v+1}\!\nmid\! m$.

\begin{theorem}\label{c char even}
Let notation be as above (specifically, ${\rm char} F=2$ hence $n$ is odd). 
The following six are equivalent to each other. 

{\rm(1)} Any self-dual $2$-quasi-cyclic code over $F$ of length $2n$ is a dihedral code.

{\rm(2)} Any self-dual $2$-quasi-cyclic code over $F$ of length $2n$ 
is a double circulant code.

{\rm(3)}  $E''=\emptyset$, i.e., 
$-1\in\langle q\rangle_{{\Bbb Z}_n^\times}$ (cf. Remark \ref{n and E'}).

{\rm(4)} Any cyclic code over $F$ of length $n$ is an LCD code.

{\rm(5)} There is no non-zero self-orthogonal cyclic code over $F$ of length $n$.

{\rm(6)} For any prime divisors $p, p'$ of $n$, 
$v(p)=v(p')\ge 1$ , where $v(p)$ and $v(p')$ are the integers such that
$2^{v(p)}\Vert\,{\rm ord}_{{\Bbb Z}_{p}^\times}(q)$ and
$2^{v(p')}\Vert\,{\rm ord}_{{\Bbb Z}_{p'}^\times}(q)$. 
\end{theorem}

\begin{proof} 
(1)$\Rightarrow$(3). 
Suppose that (3) does not hold. Then we can find $e_j\in E''$, 
hence $e_j\ne \overline e_j\in E''$. 
Set $C=(C_1\times C_2)\oplus \widehat C_{12}\le (FH)^2$, where
\begin{align*}
 &\textstyle C_1= C_2=FH e_j, \quad  
 C_{12}=FH e'~~\mbox{where}~ 
 e'=\sum_{\epsilon\in E-\{e_j,\overline e_j\}}\epsilon;\\
 &\widehat C_{12}=\{(c,c)\,|\,c\in C_{12}\}\quad 
(\mbox{i.e.,~ take $g=e'$ in Theorem \ref{Goursat C}}).
\end{align*}
Then $\overline C_{12}=C_{12}$, $g\overline g=e'e'=e'=-e'$
(as ${\rm char} F=2$), and by Corollary~\ref{c C^bot}, 
$\langle C_1,C_1\rangle_{F\!H}=0$ and
$\langle C_1,C_{12}\rangle_{F\!H}=0$.
So, by Theorem \ref{Goursat C} and Theorem \ref{self-dual C},
$C=(C_1\times C_2)\oplus \widehat C_{12}$ is a 
self-dual $2$-quasi-cyclic code over $F$ of length $2n$.
Since $\overline e_j\ne e_j$ hence $\overline C_1\ne C_2$,  
by Theorem \ref{t char even}, $C$ is not a dihedral code.
That is a contradiction to (1). 

(3)$\Rightarrow$(2). 
For any self-dual $2$-quasi-cyclic code 
$C=(C_1\times C_2)\oplus \widehat C_{12}\le (FH)^2$,
by Theorem \ref{self-dual C}(2.2), 
$\langle C_1,C_1\rangle_{F\!H}=0$; 
hence $E_{C_1}\subseteq E''=\emptyset$
(cf.~Lemma \ref{l C^bot} and its corollary); 
thus $E_{C_1}=\emptyset$; consequently, $C_1=0$. Similarly, $C_2=0$.
Then $C_{12}=FH$ (because $\dim_F C=n$).
By Theorem \ref{generator matrix}, $C$ is a double circulant code.

(2)$\Rightarrow$(1). 
If $C$ is a self-dual $2$-quasi-cyclic code, then by (2) 
$C$ is a self-dual double circulant code, hence $C$ is dihedral 
by Corollary \ref{principal self-dual}.

(3)$\Leftrightarrow$(4). It follows from Lemma \ref{l C^bot} and its corollary. 

(3)$\Leftrightarrow$(5). It follows from Lemma \ref{l C^bot} and its corollary. 

(5)$\Leftrightarrow$(6). 
Note that $n$ is odd.
Thus the equivalence is a known result,  see \cite{P92} and \cite{KR}.
\big(Note: an incomplete proof of (5)$\Leftrightarrow$(6) appeared in \cite{P92}, 
a complete proof of it was given in~\cite{KR}. 
Or, one can check (3)$\Leftrightarrow$(6) directly by 
Chinese Remainder Theorem.\big)
\end{proof}

\begin{example}\label{even examples}\rm
Take $F=F_4=\{0,1,\omega,\omega^2\}$, where $1+\omega+\omega^2=0$. 
Take $n=3$. Then the set $E$ of primitive idempotents of $FH$ is: 
$E=\{e_0,e_1,\overline e_1\}$, where $e_0=1+x+x^2$, 
$e_1=1+\omega x+\omega^2 x^2$, $\overline e_1=1+\omega^2 x+\omega x^2$. 
The irreducible ideals 
$FHe_0$, $FHe_1$ and $FH\overline e_1$ are all $1$-dimensional.

(1)~ Take $C_1\!=\!C_2=F\!H e_1$, $C_{12}\!=\!F\!H e_0$, $g=e_0$. 
Then $C\!=\!(C_1\!\times\! C_2)\oplus\widehat C_{12}\le(FH)^2$
is a special case of the code constructed in the proof of (1)$\Rightarrow$(3)
of Theorem~\ref{c char even}; 
so $C$ is a self-dual $2$-quasi-cyclic code of length~$6$ over $F$,
but it is not a dihedral code (and hence it is not a double circulant code)
because $\overline C_1\ne C_2$, cf. Theorem~\ref{2-quasi dihedral}. 
This is just Example \ref{counterexample}(1).

(2)~ If we take 
$C_1=FH e_1$, $C_2=FH \overline e_1$, $C_{12}=FH e_0$ and $g=e_0$,
then $C=(C_1\times C_2)\oplus\widehat C_{12}\le(FH)^2$
is a self-dual $2$-quasi-cyclic code, and is a dihedral code 
(cf. Theorem~\ref{t char even} and Theorem~\ref{2-quasi dihedral}).
But, by Theorem \ref{generator matrix}, 
$C$ is not a double circulant code because $C_2\ne 0$.
This is just Example \ref{counterexample}(2).

(3)~ Let $C_1=C_2=0$, $C_{12}=FH$ and 
$g=\alpha e_0+\beta e_1+\gamma\overline e_1$ 
with $\alpha,\beta,\gamma\in F$ such that $\alpha^2=\beta\gamma=1$.
Then $g\in(FH)^\times$ and $g\overline g=1=-1$ (as ${\rm char}F=2$); 
and by Theorem~\ref{generator matrix} and Theorem~\ref{self-dual C}, 
$$
C=(C_1\times C_2)\oplus\widehat C_{12}
  =\widehat C_{12}=FH(1,g)\le(FH)^2
$$
is a self-dual double circulant code with a generator matrix~$\big(I~\,g(P)\big)$ 
where $g(P)$ is defined in Eq.\eqref{FH cir};
and by Corollary \ref{principal self-dual}, it is a dihedral code.
\end{example}

\section{The case of odd characteristic}\label{odd case}

In this section, the characteristic ${\rm char}F$ is always assumed to be odd,
and $\gcd(n,q)=1$. Note that, $n$ may be even.
It is known that self-dual $2$-quasi-cyclic codes over $F$ exist if and only if
$-1$ is a square element of~$F$, if and only if $4\,|\,q-1$, see \cite{LS03}
(or \cite{FL21} for more general $2$-quasi-abelian codes).  

Instead of the dihedral group $G=H\rtimes T$, 
in this subsection we consider the group $\tilde G=H\rtimes\tilde T$, 
which is a semidirect product of the cyclic $H$ of order~$n$ by a cyclic group 
$\tilde T=\{1,\tilde y,\tilde y^2,\tilde y^3\}$ of order~$4$,
with relation $\tilde yx\tilde y^{-1}=x^{-1}$. Obviously, $|\tilde G|=4n$.
The $\tilde G$ is called a dicyclic group in literature, e.g., in~\cite{BR}. 
Obviously, $Z=\{1,\tilde y^2\}$ is a central subgroup of $\tilde G$,
and $\tilde G/Z\cong G$.  

Let $F{*}G$ be the algebra over $F$ as follows:
$F{*}G$ is the vector space with basis 
$\{1,x,\cdots,x^{n-1},~\tilde y,\, x\tilde y,\,\cdots,\,x^{n-1}\tilde y\}$, 
and the multiplication is defined by the following  relation:
\begin{equation}\label{twisted}
 x^n=1,\quad \tilde y^2=-1,\quad \tilde y x =x^{-1}\tilde y.
\end{equation}
In other words, $F{*}G=F[X,Y]/\langle X^n-1,\,Y^2+1,\, XYX-Y\rangle$, 
where $F[X,Y]$ is the polynomial algebra of two variables $X$ and $Y$
which are not commutative each other. 
In the multiplicative group $(F{*}G)^\times$, 
$x$ and $\tilde y$ generate a subgroup which is just $\tilde G$, 
but the central subgroup $Z=\{1,\tilde y^2\}$ of $\tilde G$ 
is embedded in to~$F^\times$ as $\{1,-1\}$. 
We call $F{*}G$ a {\em consta-dihedral group algebra} (cf. \cite{SR}), 
or a {\em twisted dihedral group algebra} (cf. \cite[p.268]{CR}). 
And, any left ideal $C$ of $F{*}G$ 
(i.e., $F{*}G$-submodule $C$ of the left regular module $F{*}G$)
 is called a {\em consta-dihedral code} of length $2n$ over $F$, 
and is denoted by $C\le F{*}G$.

It is still true that $FH$ is a subalgebra of $F{*}G$, and 
\begin{equation*}
 F{*}G=FH\oplus FH\tilde y=\big\{a(x)+a'(x)\tilde y\:\big|\; a(x),a'(x)\in FH \big\}.
\end{equation*}
In particular, any submodule (left ideal) of $F{*}G$ is still restricted to a 
$2$-quasi-cyclic code.
However, for any $a(x)+a'(x)\tilde y\in F{*}G$, 
\begin{equation}\label{left by tilde y}
  \tilde y\big(a(x)+a'(x)\tilde y\big)=\tilde ya(x)+\tilde ya'(x)\tilde y
  =-\overline{a'}(x)+\overline a(x)\tilde y,
\end{equation}
since $\tilde ya'(x)\tilde y=\tilde ya'(x)\tilde y^{-1}\tilde y^2
  =\overline {a'}(x)\tilde y^2=-\overline{a'}(x)$. 
Compare it with Eq.\eqref{left by y}.

Modifying Theorem \ref{2-quasi dihedral}, we can get the following
result which answers the question: what 
$2$-quasi-cyclic codes
are consta-dihedral codes?

\begin{theorem}\label{2-quasi consta-dihedral}
Let $C=(C_1\times C_2)\oplus\widehat C_{12}\le (FH)^2$ 
be a $2$-quasi-cyclic code, 
where $C_1,C_2,C_{12}\le FH$, $g\in C_{12}^\times$ 
and $\widehat C_{12}\le(FH)^2$ are as in Eq.\eqref{structure C}.
The following two are equivalent:

{\rm(1)} $C$ is a consta-dihedral code;

{\rm(2)} $\overline C_1\!=C_2$, $\overline C_{12}=C_{12}$
and $\overline g g=-e_{C_{12}}$, 
where $e_{C_{12}}$ is defined in Lemma~\ref{C e_C}.
\end{theorem}

\begin{proof} The proof is the same as the proof of Theorem \ref{2-quasi dihedral}, 
provided pay attention at two steps.
First, comparing Eq.\eqref{left by y} and Eq.\eqref{left by tilde y}, 
we see that Eq.\eqref{dihdral 1} turns into the following: 
\begin{equation*}
 c(x)+(c(x)g(x))\tilde y\in C ~~\implies~~
 -\overline{(c(x)g(x))}+\overline{c(x)}\tilde y\in C.
\end{equation*} 
which leads to $\overline g g=-e_{C_{12}}$.

Second, Eq.\eqref{dihedral 2} is changed as follows:
\begin{equation*}
\tilde y\big(c(x)+c'(x)\tilde y\big)
=-\big(\overline c_2(x) +\overline c_{12}(x)\overline g(x)\big)
 + \big(\overline c_1(x)+\overline c_{12}(x)\big)\,\tilde y.
\end{equation*}
Thus, from $\overline g g=-e_{C_{12}}$, Eq.\eqref{dihedral 3} is changed as
\begin{equation*}
-\overline c_{12}(x)\overline g(x)\cdot g(x)=
\overline c_{12}(x)\,e_{C_{12}}(x)=\overline c_{12}(x).
\end{equation*}
Then the proof would be completed.
\end{proof}

Similarly to Theorem \ref{t char even} and Corollary \ref{principal self-dual}, 
the following theorem and corollary follow at once. 

\begin{theorem}\label{t char odd}
Keep the notation as above (specifically, ${\rm char} F$ is odd).
Let $C=(C_1\times C_2)\oplus \widehat C_{12}\le(FH)^2$ in Eq.\eqref{structure C}
be a  self-dual $2$-quasi-cyclic code.
Then $C$ is a consta-dihedral code if and only if $\overline C_1=C_2$.  
\end{theorem}

\begin{corollary}\label{principal self-dual odd}
Any principal self-dual $2$-quasi-cyclic code over~$F$ 
(with ${\rm char} F$ odd) is consta-dihedral.
In particular, any self-dual double circulant code over~$F$ 
(with ${\rm char} F$ odd) is consta-dihedral.
\end{corollary}

\begin{theorem}\label{c char odd}
Let notation be as above (specifically, ${\rm char} F$ is odd). 
The following six are equivalent to each other:

{\rm(1)} Any self-dual $2$-quasi-cyclic code over $F$ of length $2n$ 
is a consta-dihedral code.

{\rm(2)} Any self-dual $2$-quasi-cyclic code over $F$ of length $2n$ 
is a double circulant code.

{\rm(3)}  $E''=\emptyset$, i.e., 
$-1\in\langle q\rangle_{{\Bbb Z}_n^\times}$ (cf. Remark \ref{n and E'}).

{\rm(4)} Any cyclic code over $F$ of length $n$ is an LCD code.

{\rm(5)} There is no non-zero self-orthogonal cyclic code over $F$ of length $n$.

{\rm(6)} 
If $4\nmid n$ then, for any odd prime divisors $p$, $p'$ of $n$, 
$v(p)=v(p')\ge 1$, where $v(p)$ and $v(p')$ are as the same as in
Theorem \ref{c char even}.
Otherwise, $q\equiv -1\pmod{2^\ell}$ where $2^\ell\Vert n$, and
for any odd prime divisor $p$ of $n$, $v(p)=1$.
\end{theorem}

\begin{proof}
The proof of the equivalences for (1), (2), (3), (4), (5) 
are the same as the proof of Theorem \ref{c char even}. 

(5)$\Leftrightarrow$(6). 
The equivalence follows from a known result, see \cite{L PhD} or~\cite{LLC}.
\big(Note: because here $n$ may be even, the result in \cite{KR}
is still incomplete; \cite{L PhD} and \cite{LLC} obtain the complete correct
result. Compare it with the proof of (5)$\Leftrightarrow$(6) 
of Theorem~\ref{c char even}.\big)
\end{proof}

\begin{example}\label{odd examples}\rm
Take $F=F_5=\{0,\pm1,\pm2\}$. 
Take $n=4$. Then the set $E$ of primitive idempotents of $FH$ is: 
$E=\{e_0,e_1,e_2,\overline e_2\}$, 
where $e_0=-1-x-x^2-x^3$, 
$e_1=-1+ x- x^2+x^3$, 
$e_2=-1+2 x+ x^2-2x^3$, 
$\overline e_2=-1-2 x+x^2+2x^3$.
The irreducible ideals of $FH$ are all $1$-dimensional.
Set $e=e_0+e_1$. Then $\dim_F FHe=2$, $\overline e=e$.

(1) Take $C_1\!=\!C_2\!=\!FH e_2$, $C_{12}\!=\!FH e$, $g\!=\!2e$.
Then $g\overline g=2e 2e=-e$.
By Theorem \ref{self-dual C} 
(the other conditions in the theorem can be checked by Corollary~\ref{c C^bot}), 
$C=(C_1\times C_2)\oplus\widehat C_{12}\le(FH)^2$
is a self-dual $2$-quasi-cyclic code over $F$ of length~$8$.
But, by Theorem \ref{2-quasi consta-dihedral}, $C$ is not a consta-dihedral code
because $\overline C_1\ne C_2$.
And, by Theorem \ref{generator matrix}, 
$C$ is not a double circulant code because $C_2\ne 0$.

(2) If we take 
$C_1=FH e_2$, $C_2=FH \overline e_2$, $C_{12}=FH e$ and $g=2e$,
then $g\overline g=-e$, and by Theorem \ref{self-dual C},
$C=(C_1\times C_2)\oplus\widehat C_{12}\le(FH)^2$
is a self-dual $2$-quasi-cyclic code; 
and by Theorem \ref{2-quasi consta-dihedral}, $C$ is a consta-dihedral code.
But, by Theorem \ref{generator matrix}, 
$C$ is not a double circulant code because $C_2\ne 0$.

(3) Let $C_1=C_2=0$, $C_{12}=FH$ and 
$g=\alpha e_0+\beta e_1+\gamma e_2+\delta\overline e_2$ 
with $\alpha,\beta,\gamma,\delta\in F$ such that 
$\alpha^2=\beta^2=\gamma\delta=-1$.
Then $g\in (FH)^\times$ and $g\overline g=-1$, 
and $C=(C_1\times C_2)\oplus\widehat C_{12}=\widehat C_{12}=FH(1,g)\le(FH)^2$
is a self-dual $2$-quasi-cyclic code,  is a consta-dihedral code,
and has a generator matrix~$\big(I~\,g(P)\big)$ where 
$g(P)$ is defined in Eq.\eqref{FH cir}.
\end{example}

We conclude the paper by few remarks.

\begin{remark}\rm
(1) The main contribution of this paper is that we characterized 
the structure of $2$-quasi-cyclic codes over finite fields by Goursat Lemma; 
and with the characterization, we described the exact relationships
between the three classes of codes:  (self-dual) $2$-quasi-cyclic codes, 
(self-dual) double circulant codes, and (self-dual) dihedral codes.   
Such a method seems powerful for some studies on $2$-quasi-cyclic codes.   

(2) Many arguments in this paper are relied on the ``bar'' map $a\to\overline a$.
In fact, as polynomial $a=a(x)\in FH$, $\overline a(x)$ is related to
the so-called {\em reciprocal polynomial} $a^*(x)$ very closely; 
at least, they have roots in common. An advantage of $\overline a(x)$ is that
the ``bar'' map is an automorphism of the $F$-algebra~$FH$.  

(3) For the consta-dihedral group algebra $F{*}G$, 
we can put the group $\tilde G$ aside and define $F{*}G$ 
by Eq.\eqref{twisted} directly over any finite field $F$
(${\rm char}F=2$ or not). 
Then we can identify $F{*}G$ 
with the dihedral group algebra $FG$ once ${\rm char}F=2$, 
because $\tilde y^2=-1=1$ once ${\rm char}F=2$.
Up to this unification, 
Theorem \ref{t char odd} and Theorem \ref{c char odd}
(for any finite field), respectively,
cover Theorem~\ref{t char even} and Theorem~\ref{c char even}, respectively.

(4) For Theorem \ref{c char odd} 
(for any finite field in the sense of the remark (3) above), we remark two points.
(i) The equivalences each other of (3), (4), (5), (6) of the theorem are known, where
the equivalences between (3), (4), (5) are clear by Corollary \ref{c C^bot}, 
and the equivalence of (5) and (6) was undergoing from \cite{P92}, \cite{KR}
to \cite{L PhD}, \cite{LLC}. 
(ii) Our contribution is that we found the equivalence of (1), (2) 
to the others. As a comparison, (1), (2) are about $2$-quasi-cyclic codes,
whereas (3), (4), (5), (6) are about only cyclic codes.   
\end{remark}

\section*{Acknowledgements}
The research of the first author is supported 
by NSFC with grant number 12171289.

\end{document}